\documentclass[11pt,letterpaper]{article}

\usepackage{amsmath,amsfonts,amsthm,amssymb,stmaryrd,graphicx,relsize,bm}  % 9.20 add graphicx,resize
\theoremstyle{plain}

\numberwithin{equation}{section}

\newtheorem{thm}{Theorem}[section]
\newtheorem{lem}[thm]{Lemma}
\newtheorem{cor}[thm]{Corollary}

\newenvironment{exam}[1]
{\begin{flushleft}\textbf{Example #1}.\enspace}%
{\end{flushleft}}

{\end{flushleft}}

\allowdisplaybreaks  % introduced 7.15.15

\newcommand{\real}{{\mathbb R}}

\newcommand{\tbullet}{\mathrel{\raise .4ex\hbox{\tiny$\bullet$}}} % 5.8.20 THIS for larger cdot as times

\newcommand{\trace}{tr}
\newcommand{\rmcor}{\mathrm{Cor}}

\newcommand{\ityes}{\textit{yes}}
\newcommand{\itno}{\textit{no}}
\newcommand{\rmre}{\mathrm{Re}}

\newcommand{\escript}{\mathcal{E}}

\newcommand{\hscript}{\mathcal{H}}
\newcommand{\iscript}{\mathcal{I}}
\newcommand{\jscript}{\mathcal{J}}
\newcommand{\lscript}{\mathcal{L}}

\newcommand{\sscript}{\mathcal{S}}

\newcommand{\atilde}{\widetilde{A}}
\newcommand{\btilde}{\widetilde{B}}
\newcommand{\ctilde}{\widetilde{C}}

\newcommand{\rhotilde}{\widetilde{\rho}}

\newcommand{\iscriptbar}{\overline{\iscript}}
\newcommand{\jscriptbar}{\overline{\jscript}}

\newcommand{\ab}[1]{\left|#1\right|}
\newcommand{\brac}[1]{\left\{#1\right\}}
\newcommand{\paren}[1]{\left(#1\right)}
\newcommand{\sqbrac}[1]{\left[#1\right]}
\newcommand{\elbows}[1]{{\left\langle#1\right\rangle}}
\newcommand{\ket}[1]{{\left|#1\right>}}
\newcommand{\bra}[1]{{\left<#1\right|}}

\begin{document}

\title{CONDITIONED EFFECTS,\\OBSERVABLES AND INSTRUMENTS}
\author{Stan Gudder\\ Department of Mathematics\\
University of Denver\\ Denver, Colorado 80208\\
sgudder@du.edu}
\date{}
\maketitle

\begin{abstract}
We begin with a study of operations and the effects they measure. We define the probability that an effect $a$ occurs when the system is in a state $\rho$ by
$P_\rho (a)=\trace (\rho a)$. If $P_\rho (a)\ne 0$ and $\iscript$ is an operation that measures $a$, we define the conditional probability of an effect $b$ given $a$ relative to $\iscript$ by
\begin{equation*}
P_\rho (b\mid a)=\trace\sqbrac{\iscript (\rho )b}/P_\rho (a)
\end{equation*}
We characterize when Bayes' quantum second rule
\begin{equation*}
P_\rho (b\mid a)=\frac{P_\rho (b)}{P_\rho (a)}\,P_\rho (a\mid b)
\end{equation*}
holds. We then consider L\"uders and Holevo operations. We next discuss instruments and the observables they measure. If $A$ and $B$ are observables and an instrument $\iscript$ measures $A$, we define the observable $B$ conditioned on $A$ relative to $\iscript$ and denote it by $(B\mid A)$. Using these concepts, we introduce Bayes' quantum first rule. We observe that this is the same as the classical Bayes' first rule, except it depends on the instrument used to measure $A$. We then extend this to Bayes' quantum first rule for expectations. We show that two observables $B$ and $C$ are jointly commuting if and only if there exists an atomic observable $A$ such that $B=(B\mid A)$ and $C=(C\mid A)$. We next obtain a general uncertainty principle for conditioned observables. Finally, we discuss observable conditioned quantum entropies. The theory is illustrated with many examples.
\end{abstract}

\section{Effects and Operations}  % Section 1
It is sometimes stated that all probabilities in quantum mechanics are conditional probabilities and there is some sense to this statement. Underlying most quantum experiments or observations, there are basic observables $A_i$ and calculations are performed according to the outcomes obtained for $A_i$. For example, many quantum experiments consist of scattered particles and these involve the positions $A_i$ of the various particles. The probabilities for another observable is thus conditioned by the outcomes of $A_i$.

According to complexity, there is a hierarchy of quantum measurements. The simplest are effects, the next are observables and finally we have instruments. Each of these types of measurements can be conditioned in a systematic way. They can even be conditioned among each other.

Let $H$ be a finite-dimensional complex Hilbert space representing a quantum system. The set of linear operators on $H$ is denoted by $\lscript (H)$ and the set of self-adjoint operators is denoted by $\lscript _S(H)$. A \textit{state} is a positive operator $\rho\in\lscript _S(H)$ with trace $\trace (\rho )=1$ and the set of states is denoted by $\sscript (H)$. States describe the conditions of the system and are employed to compute probabilities of measurement outcomes. An operator $a$ satisfying $0\le a\le I$ is called an \textit{effect}. An effect represents a two outcome \ityes-\itno\ experiment that either occurs or does not occur
\cite{bgl95,dl70,hz12,kra83,nc00}. We represent the set of effects by $\escript (H)$. If $a\in\escript (H)$ occurs, then its \textit{complement} $a'=I-a$ does not occur.
An \textit{operation} is a completely positive linear map $\iscript\colon\lscript (H)\to\lscript (H)$ such that $\trace\sqbrac{\iscript (\rho )}\le\trace (\rho )$ for all
$\rho\in\sscript (H)$ \cite{bgl95,dl70,hz12,kra83,nc00}. An operation that satisfies $\trace\sqbrac{\iscript (\rho )}=\trace (\rho )$ for all $\rho\in\sscript (H)$ is called a
\textit{channel} \cite{hz12,hol94,nc00}. Any operation $\iscript$ has a \textit{Kraus decomposition} $\iscript (A)=\sum\limits _iK_iAK_i^*$ where $K_i\in\lscript (H)$ and
$\sum\limits _iK_i^*K_i\le I$. We call $K_i$, $i=1,2,\ldots ,n$, \textit{Kraus operators} for $\iscript$ \cite{kra83}. When $\iscript$ is a channel, we have
$\sum\limits _iK_i^*K_i=I$.

Corresponding to an operation $\iscript$ we have the \textit{dual operation} \cite{gud21,gud223,gud320} $\iscript ^*\colon\lscript (H)\to\lscript (H)$ where $\iscript ^*$ is linear and satisfies $\trace\sqbrac{\iscript (\rho )A}=\trace\sqbrac{\rho\iscript ^*(A)}$ for all $\rho\in\sscript (H)$, $A\in\lscript (H)$. If $\iscript$ has Kraus decomposition
$\iscript (A)=\sum K_iAK_i^*$, then $\iscript ^*(A)=\sum K_i^*AK_i$ for all $A\in\lscript (H)$. If $\iscript$ is a channel then $\iscript ^*(I)=I$. It is easy to check that if
$\iscript$ is an operation, then $\iscript ^*\colon\escript (H)\to\escript (H)$ and $\iscript ^*(a)\le a$ for all $a\in\escript (H)$. We say that an operation $\iscript$
\textit{measures an effect} $a$ if $\trace\sqbrac{\iscript (\rho )}=\trace (\rho a)$ for all $\rho\in\iscript (H)$ \cite{gud220,gud21,gud223}. We interpret
$P_\rho (a)=\trace (\rho a)$ as the probability that $a$ occurs when the system is in state $\rho$. It follows that an operation measures a unique effect. However, as we shall see, there are many operations that measure an effect $a$. If $\iscript$ measures $a$, then
\begin{equation*}
\trace\sqbrac{\rho\iscript ^*(I)}=\trace\sqbrac{\iscript (\rho )}=\trace (\rho a)
\end{equation*}
for every $\rho\in\sscript (H)$. Hence, $\iscript$ measures $a$ if and only if $\iscript ^*(I)=a$.

If $a,b\in\escript (H)$ we write $a\perp b$ if $a+b\in\escript (H)$. If $a,b\in\escript (H)$ and $\iscript$ measures $a$, we define the $\iscript$-\textit{sequential product of}
$a$ \textit{then} $b$ by $a\sqbrac{\iscript}b=\iscript ^*(b)$. It is easy to check that $a\sqbrac{\iscript}b\le a$, if $b\perp c$ then
$a\sqbrac{\iscript}(b+c)=a\sqbrac{\iscript}b+a\sqbrac{\iscript}c$ and $a\sqbrac{\iscript}I=a$ \cite{gud21,gud223}. An effect $a$ is \textit{sharp} if $a$ is a projection and $a$ is \textit{atomic} if $a$ is a one-dimensional projection.

\begin{lem}    % Lemma 1.1
\label{lem11}
Let $\iscript$ be an operation that measures $a\in\escript (H)$.
{\rm{(i)}}\enspace $\iscript ^*(b)\le a$ for all $b\in\escript (H)$.
{\rm{(ii)}}\enspace If $a$ is sharp, then $\iscript ^*(b)a=a\iscript ^*(b)$ for all $b\in\escript (H)$.
{\rm{(iii)}}\enspace If $a$ is atomic, then $\iscript ^*(b)=\lambda a$ for some $\lambda\in\sqbrac{0,1}$.
\end{lem}
\begin{proof}
(i)\enspace Since
\begin{equation*}
\iscript ^*(b)+\iscript ^*(b')=\iscript ^*(b+b')=\iscript ^*(I)=a
\end{equation*}
we conclude that $\iscript ^*(b)\le a$.
(ii)\enspace Since $\iscript ^*(b)\le a$, $\iscript ^*(b)$ and $a$ coexist \cite{hz12}. Then $a$ being sharp implies that $\iscript ^*(b)a=a\iscript ^*(b)$.
(iii)\enspace If $a$ is atomic and $\iscript ^*(b)\le a$, we have that $\iscript ^*(b)=\lambda a$ for some $\lambda\in\sqbrac{0,1}$ \cite{hz12}.
\end{proof}

If $P_\rho (a)\ne 0$ and $\iscript$ measures $a$, we define the \textit{conditional probability of} $b$ \textit{given} $a$ \textit{relative to} $\iscript$ by \cite{gud320}
\begin{equation*}
P_\rho (b\mid a)=\frac{\trace\sqbrac{\iscript (\rho )b}}{P_\rho (a)}
\end{equation*}
We then have
\begin{align*}
P_\rho (b\mid a)&=\frac{\trace\sqbrac{\iscript (\rho )b}}{\trace\sqbrac{\iscript (\rho )}}=\frac{\trace\sqbrac{\rho\iscript ^*(b)}}{\trace (\rho a)}
   =\frac{\trace\paren{\rho a\sqbrac{\iscript}b}}{\trace (\rho a)}=\frac{P_\rho\paren{a\sqbrac{\iscript}b}}{P_\rho (a)}\\
   \noalign{\smallskip}
   &=\frac{P_{\iscript (\rho )}(b)}{P_\rho (a)}
\end{align*}
We have that $b\mapsto P_\rho (b\mid a)$ is a probability distribution in the sense that $P_\rho (I\mid a)=1$ and if $b_i\in\escript (H)$ with
$b_1+b_2+\cdots +b_n\le I$, then
\begin{equation*}
P_\rho\paren{\sum _{i=1}^nb_i\mid a}=\sum _{i=1}^nP_\rho (b_i\mid a)
\end{equation*}
We also see that $\rhotilde =\iscript (\rho )/P_\rho (a)$ is a state called the \textit{updated state} for $\iscript$ and we have
\begin{equation*}
P_\rho (b\mid a)=\trace (\,\rhotilde b)=P_{\rhotilde}(b)
\end{equation*}
Thus, to find $P_\rho (b\mid a)$ we first measure $a$ using $\iscript$, update the state to $\rhotilde$ and then compute the probability of $b$ using $\rhotilde$. If
$\iscript$ and $\jscript$ are operations, we define the \textit{sequential product of} $\iscript$ \textit{then} $\jscript$ as the operation given by 
$(\iscript\circ\jscript )(\rho )=\jscript\paren{\iscript (\rho )}$ for all $\rho\in\sscript (H)$ \cite{gud21,gud223}. In a similar way we define
$(\iscript ^*\circ\jscript ^*)(A)=\jscript ^*\paren{\iscript ^*(A)}$.

\begin{thm}    % Theorem 1.2
\label{thm12}
Let $\iscript$ and $\jscript$ be operations.
{\rm{(i)}}\enspace $(\iscript\circ\jscript )^*=(\jscript ^*\circ\iscript ^*)$.
{\rm{(ii)}}\enspace If $\iscript$ measures $a$ and $\jscript$ measures $b$, the $\iscript\circ\jscript$ measures $a\sqbrac{\iscript}b$.
{\rm{(iii)}}\enspace If $a$ is measured with $\iscript$, $b$ with $\jscript$ and $a\sqbrac{\iscript}b$ with $\iscript\circ\jscript$, then
\begin{equation*}
a\sqbrac{\iscript}\paren{b\sqbrac{\jscript}c}=\paren{a\sqbrac{\iscript}b}\sqbrac{\iscript\circ\jscript}c
\end{equation*}
{\rm{(iv)}}\enspace For all $\rho\in\sscript (H)$ we have
\begin{equation*}
\trace (\rho a)P_\rho\paren{b\sqbrac{\jscript}c\mid a}=\trace\paren{\rho a\sqbrac{\iscript}b}P_\rho\paren{c\mid a\sqbrac{\iscript}b}
\end{equation*}
\end{thm}
\begin{proof}
(i)\enspace For all $\rho\in\sscript (H)$ $A\in\lscript (H)$ we obtain
\begin{align*}
\trace\sqbrac{\rho (\iscript\circ\jscript )^*(A)}&=\trace\sqbrac{(\iscript\circ\jscript )(\rho )A}=\trace\sqbrac{\jscript\paren{\iscript (\rho )}A}
   =\trace\sqbrac{\iscript (\rho )\jscript ^*(A)}\\
   &=\trace\sqbrac{\rho\iscript ^*\paren{\jscript ^*(A)}}=\trace\sqbrac{\rho (\jscript ^*\circ\iscript ^*)(A)}
\end{align*}
It follows that $(\iscript\circ\jscript )^*=\jscript ^*\circ\iscript ^*$.
(ii)\enspace Since
\begin{equation*}
\trace\sqbrac{\iscript\circ\jscript (\rho )}=\trace\sqbrac{\jscript\paren{\iscript (\rho )}}=\trace\sqbrac{\iscript (\rho )b}=\trace\sqbrac{\rho\iscript ^*(b)}
   =\trace\paren{\rho a\sqbrac{\iscript}b}
\end{equation*}
it follows that $\iscript\circ\jscript$ measures $a\sqbrac{\iscript}b$.
(iii)\enspace Applying (i) gives
\begin{align*}
a\sqbrac{\iscript}\paren{b\sqbrac{\jscript}c}&=a\sqbrac{\iscript}\paren{\jscript ^*(c)}=\iscript ^*\paren{\jscript ^*(c)}=\jscript ^*\circ\iscript ^*(c)\\
   &=(\iscript\circ\jscript )^*(c)=\paren{a\sqbrac{\iscript}b}\sqbrac{\iscript\circ\jscript}c
\end{align*}
(iv)\enspace This follows from
\begin{align*}
\trace (\rho a)P_\rho\paren{b\sqbrac{\jscript}c\mid a}&=\trace (\rho a)\frac{\trace\paren{\iscript (\rho )b\sqbrac{\jscript}c}}{\trace (\rho a)}
   =\trace\sqbrac{\iscript (\rho )\jscript ^*(c)}\\
   &=\trace\sqbrac{\jscript\paren{\iscript (\rho )}c}=\trace\sqbrac{(\iscript\circ\jscript )(\rho )c}]\\
   &=\trace\sqbrac{\rho (\iscript\circ\jscript )^*}=P_\rho\paren{a\sqbrac{\iscript}b}P_\rho\paren{c\mid a\sqbrac{\iscript}b}\qedhere
\end{align*}
\end{proof}

\textit{Bayes' second rule} says that
\begin{equation}                % equation (1.1)
\label{eq11}
P_\rho (b\mid a)=\frac{P_\rho (b)}{P_\rho (a)}\,P_\rho (a\mid b)
\end{equation}
The following lemma shows that this result does not always hold.

\begin{lem}    % Lemma 1.3
\label{lem13}
The following statements are equivalent.
{\rm{(i)}}\enspace Equation~\eqref{eq11} holds.
{\rm{(ii)}}\enspace Whenever $\iscript$ measures $a$ and $\jscript$ measures $b$, then
\begin{equation*}
\trace\paren{\rho a\sqbrac{\iscript}b}=\trace\paren{\rho b\sqbrac{\jscript}a}
\end{equation*}
{\rm{(iii)}}\enspace Whenever $\iscript$ measures $a$ and $\jscript$ measures $b$, then $\trace\sqbrac{\iscript (\rho )b}=\trace\sqbrac{\jscript (\rho )a}$.
\end{lem}
\begin{proof}
(i)$\Rightarrow$(ii) If (i) holds, then
\begin{align*}
\trace\paren{\rho a\sqbrac{\iscript}b}&=\trace\sqbrac{\rho\iscript ^*(b)}=\trace\sqbrac{\iscript (\rho )b}=P_\rho (a)P_\rho (b\mid a)=P_\rho (b)P_\rho (a\mid b)\\
   &=\trace\sqbrac{\jscript (\rho )a}=\trace\sqbrac{\rho\jscript ^*(a)}=\trace\paren{\rho b\sqbrac{\jscript}a}
\end{align*}
Hence, (ii) holds. (ii)$\Rightarrow$(iii). If (ii) holds, then
\begin{align*}
\trace\sqbrac{\iscript (\rho )b}&=\trace\sqbrac{b\iscript ^*(b)}=\trace\paren{\rho a\sqbrac{\iscript}b}=\trace\paren{\rho b\sqbrac{\jscript}a}\\
   &=\trace\sqbrac{\rho\jscript ^*(a)}=\trace\sqbrac{\jscript (\rho )a}
\end{align*}
Hence, (iii) holds. (iii)$\Rightarrow$(i) If (iii) holds then
\begin{equation*}
P_\rho (b\mid a)=\frac{\trace\sqbrac{\iscript (\rho )b}}{P_\rho (a)}=\frac{\trace\sqbrac{\jscript (\rho )a}}{P_\rho (a)}=\frac{P_\rho (b)\trace (a\mid b)}{P_\rho (a)}
\end{equation*}
Hence, (i) holds.
\end{proof}

\begin{cor}    % Corollary 1.4
\label{cor14}
.If $\iscript$ measures $a$ and $\jscript$ measures $b$, then the following statements are equivalent.
(i)\enspace \eqref{eq11} holds for every $\rho\in\sscript (H)$.
(ii)\enspace $a\sqbrac{\iscript}b=b\sqbrac{\jscript}a$.
(iii)\enspace $\iscript ^*(b)=\jscript ^*(a)$.
\end{cor}

\begin{exam}{1}  % Example 1
For $a\in\escript (H)$ we define the L\"uders operation $\lscript ^{(a)}(\rho )=a^{1/2}\rho a^{1/2}$ \cite{dl70, lud51}. Then
\begin{equation*}
\trace\sqbrac{\lscript ^{(a)}(\rho )}=\trace (a^{1/2}\rho a^{1/2})=\trace (\rho a)
\end{equation*}
for all $\rho\in\sscript (H)$ so $\lscript ^{(a)}$ measures $a$. Notice that $\lscript ^{(a)*}=\lscript ^{(a)}$ for all $a\in\escript (H)$ and
$a\sqbrac{\lscript ^{(a)}}b=a^{1/2}ba^{1/2}$. We call $a\sqbrac{\lscript ^{(a)}}b$ the \textit{standard sequential product of} $a$ \textit{then} $b$
\cite{gg02,gud21}. Relative to $\lscript ^{(a)}$ we have for all $\rho\in\sscript (H)$, $b\in\escript (H)$ that
\begin{equation*}
P_\rho (a\mid b)=\frac{\trace\sqbrac{\lscript ^{(a)}(\rho )b}}{P_\rho (a)}=\frac{\trace (a^{1/2}\rho a^{1/2}b)}{\trace (\rho a)}
   =\frac{\trace (\rho a^{1/2}ba^{1/2})}{\trace{\rho a)}}
\end{equation*}
Applying Corollary~\ref{cor14} we have that Bayes' second rule holds relative to $\lscript ^{(a)}$ and $\lscript ^{(b)}$ for all $\rho\in\sscript (H)$ if and only if
$a^{1/2}ba^{1/2}=b^{1/2}ab^{1/2}$. This is equivalent to $ab=ba$; that is, $a$ and $b$ commute \cite{gg02}. Thus, \eqref{eq11} does not hold, in general. We also have from Theorem~\ref{thm12}(iii) that 
\begin{equation*}
a\sqbrac{\lscript ^a}\paren{b\sqbrac{\lscript {(b)}}c}=\paren{a\sqbrac{\lscript ^{(a)}b}\sqbrac{\lscript ^{(a)}\circ\lscript ^{(b)}}c}=a^{1/2}b^{1/2}cb^{1/2}a^{1/2}
\end{equation*}
It follows from Theorem~\ref{thm12}(ii) that $\lscript ^{(a)}\circ\lscript ^{(b)}$ measures $a\sqbrac{\lscript ^{(a)}}b=a^{1/2}ba^{1/2}$. However,
\begin{align*}
(\lscript ^{(a)}\circ\lscript ^{(b)})(\rho )&=\lscript ^{(b)}(\lscript ^{(a)}(\rho ))=\lscript ^{(b)}(a^{1/2}\rho a^{1/2})=b^{1/2}a^{1/2}\rho a^{1/2}b^{1/2}
\intertext{and}
\lscript ^{(a^{1/2}ba^{1/2})}(\rho )&=(a^{1/2}ba^{1/2})^{1/2}\rho (a^{1/2}ba^{1/2})^{1/2}
\end{align*}
so $\lscript ^{(a)}\circ\lscript ^{(b)}\ne\lscript ^{\paren{a\sqbrac{\lscript ^{(a)}}b}}$. We conclude that
\begin{equation*}
a\sqbrac{\lscript ^{(a)}}\paren{b\sqbrac{\lscript ^{(b)}}c}\ne\paren{a\sqbrac{\lscript ^{(a)}}b}\sqbrac{\lscript ^{\paren{a\sqbrac{\lscript ^{(a)}}b}}}c
\end{equation*}
in general.\hfill\qedsymbol
\end{exam}

\begin{exam}{2}  % Example 2
If $a\in\escript (H)$, $\alpha\in\sscript (H)$, we define the \textit{Holero operation} \cite{hol82,hol94}
\begin{equation*}
\hscript ^{(a,\alpha )}(\rho )=\trace (\rho a)\alpha
\end{equation*}
Then for every $\rho\in\sscript (H)$, $b\in\escript (H)$ we obtain
\begin{align*}
\trace\sqbrac{\rho\hscript ^{(a.\alpha )*}(b)}&=\trace\sqbrac{\hscript ^{(a,\alpha )}(\rho )b}=\trace\sqbrac{\trace (\rho a)\alpha b}\\
     &=\trace (\rho a)\trace (\alpha b)=\trace\sqbrac{\rho\trace (\alpha b)a}
\end{align*}
Hence,
\begin{equation*}
\hscript ^{(a,\alpha )*}(b)=a\sqbrac{\hscript ^{(a,\alpha )}}b=\trace (\alpha b)a
\end{equation*}
Since $\trace\sqbrac{\hscript ^{(a,\alpha )}(\rho)}=\trace (\rho a)$ we see that $\hscript ^{(a,\alpha )}$ measures $a$. This shows that for any $a\in\escript (H)$, there are many operations that measure $a$. The conditional probability of $b$ given $a$ related to $\hscript ^{(a,\alpha )}$ becomes 
\begin{equation*}
P_\rho (b\mid a)=\frac{\trace\sqbrac{\hscript ^{(a,\alpha )}(\rho )b}}{P_\rho (a)}=\frac{\trace (\rho a)\trace (\alpha b)}{\trace (\rho a)}=\trace (\alpha b)
\end{equation*}
which curiously is independent of $\rho$ and $a$. Applying Corollary~\ref{cor14} we have that Bayes' second rule holds for all $\rho\in\sscript (H)$ relative to
$\hscript ^{(a,\alpha )}$ and $\hscript ^{(b,\beta )}$ if and only if
\begin{equation*}
\trace (\alpha b)a=\trace (\beta a)b
\end{equation*}
If $a$ and $b$ are sharp this is equivalent to $a=b$ and $\trace (\alpha a)=\trace (\beta a)$. Moreover, Theorem~\ref{thm12}(iii) becomes
\begin{align*}
a\sqbrac{\hscript ^{(a,\alpha )}}\paren{b\sqbrac{\hscript ^{(b,\beta )}}c}
   &=\paren{a\sqbrac{\hscript ^{(a,\alpha )}}b}\sqbrac{\hscript ^{(a,\alpha )}\circ\hscript ^{(b,\beta )}}c\\
   &=a\sqbrac{\hscript ^{(a,\alpha )}}\paren{\hscript ^{(b,\beta )*}(c)}=a\sqbrac{\hscript ^{(a,\alpha )}}\paren{\trace (\beta c)b}\\
   &=\trace (\beta c)a\sqbrac{\hscript ^{(a,\alpha )}}b=\trace (\beta c)\hscript ^{(a,\alpha )*}(b)\\
   &=\trace (\beta c)\trace (\alpha b)a
\end{align*}
Unlike the L\"uders operations, we have
\begin{equation*}
\hscript ^{(a,\alpha )}\circ\hscript ^{(b,\beta )}=\hscript ^{\paren{a\sqbrac{\hscript ^{(a,\alpha )}}b,\beta}}
\end{equation*}
Indeed,
\begin{align*}
\hscript ^{(a,\alpha )}\circ\hscript ^{(b,\beta )}(\rho )&=\hscript ^{(b,\beta )}\sqbrac{\hscript ^{(a,\alpha )}(\rho )}=\hscript ^{(b,\beta )}\paren{\trace (\rho a)\alpha}\\
    &=\trace (\rho a)\hscript ^{(b,\beta )}(\alpha )=\trace (\rho a)\trace (\alpha b)\beta\\
    &=\trace\sqbrac{\rho\trace (\alpha b)a}\beta =\hscript ^{\paren{\trace (\alpha b)a,\beta}}(\rho )\\
    &=\hscript ^{\paren{\hscript (a,\alpha )^*(b),\beta}}(\rho )=\hscript ^{\paren{a\sqbrac{\hscript ^{(a,\alpha )}}b,\beta}}(\rho)\hskip 6pc\qedsymbol
\end{align*}
\end{exam}

\section{Observables and Instruments}  % Section 2
A (finite) \textit{observable} is a collection of effects $A=\brac{A_x\colon x\in\Omega _A}$ on $H$ satisfying $\sum\limits _{x\in\Omega _A}A_x=I$
\cite{bgl95,dl70,hz12,nc00}. We assume that the set $\Omega _A$ is finite and call $\Omega _A$ the \textit{outcome space} for $A$. We think of $A$ as an experiment or measurement and when the outcome $x$ results, then we say that the effect $A_x$ \textit{occurs}. The condition $\sum\limits _{x\in\Omega _A}A_x=I$ means that one of the outcomes occurs when a measurement of $A$ is performed. If $\rho\in\sscript (H)$, then $P_\rho (A_x)=\trace (\rho A_x)$ is the probability that the outcome $x$ results and $A_x$ occurs. We call $A(\Delta )=\sum\brac{A_x\colon x\in\Delta}$, where $\Delta\subseteq\Omega _A$, a \textit{positive operator-valued measure} (POVM). The \textit{probability distribution} of $A$ in the state $\rho$ is the measure given by $\Phi _\rho ^\Delta (\Delta )=\sum\limits _{x\in\Delta}P_\rho (x)$ for all
$\Delta\in\Omega _A$ and we usually write
\begin{equation*}
\Phi _\rho ^\Delta (x)=\Phi _\rho ^\Delta\paren{\brac{x}}=P_\rho (A_x)
\end{equation*}

A (finite) \textit{instrument} is a finite collection of operations $\iscript =\brac{\iscript _x\colon x\in\Omega _\iscript}$ such that
$\iscriptbar=\sum\limits _{x\in\Omega _\iscript}\iscript _x$ is a channel \cite{bgl95,dl70,gud120,hz12,nc00}. Then for all $\rho\in\sscript (H)$ and
$\Delta\subseteq\Omega _\iscript$
\begin{equation*}
\Phi _\rho ^\iscript (\Delta )=\sum\brac{\trace\sqbrac{\iscript _x(\rho )}\colon x\in\Delta}
\end{equation*}
is the probability measure on $\Omega _\iscript$. We say that $\iscript$ \textit{measures} an observable $A$ if for all $\rho\in\sscript (H)$, we have
$\trace\sqbrac{\iscript _x(\rho )}=\trace (\rho A_x)$ for every $x\in\Omega _A$. Clearly, $\iscript$ measures a unique observable and they both have the same probability distribution. As with operations and effects an observable is measured by many instruments. If $\iscript$ is an instrument, its \textit{dual instrument}
$\iscript ^*\colon\lscript (H)\to\lscript (H)$ satisfies \cite{gud220,gud223}
\begin{equation*}
\trace\sqbrac{\rho\iscript _x^*(A)}=\trace\sqbrac{\iscript _x(\rho)A}
\end{equation*}
for all $A\in\lscript (H)$
\begin{equation*}
\iscript _{\Omega _\iscript}^*(I)=\sum _{x\in\Omega _\iscript}\iscript _x^*(I)=I
\end{equation*}
It is easy to check that $\iscript _x^*\colon\escript (H)\to\escript (H)$ and $\iscript _x^*(I)$ is the observable measured by $\iscript$.

If $a\in\escript (H)$ and $A$ is an observable on $H$ measured by the instrument $\iscript$, the effect $a$ \textit{conditioned by} $A$ is the effect
\begin{equation*}
(a\mid A)=\sum _{x\in\Omega _A}\iscript _x^*(a)=\iscript _{\Omega _A}^*(a)=\sum _{x\in\Omega _A}A_x\sqbrac{\iscript _x}a
\end{equation*}
It is clear that $a\mapsto (a\mid A)$ is a morphism in the sense that $(I\mid A)=I$ and if $a_i\in\escript (H)$ with $\sum\limits _{i=1}^na_i\le I$ then
$\paren{\sum\limits _{i=1}^na_i\mid A}=\sum\limits _{i=1}^n (a_i\mid A)$. A \textit{sub-observable} is a finite collection of effects $A=\brac{A_x\colon x\in\Omega _A}$ on $H$ satisfying $\sum\limits _{x\in\Omega _A}A_x\le I$ \cite{gud223}. If $A$ is a sub-observable, then $A$ possesses a unique minimal extension to an observable by adjoining the effect $I-\sum\limits _{x\in\Omega _A}A_x$ to $A$. If $A$ is an observable and $a\in\escript (H)$ is measured by an operator $\iscript$, then $A$
\textit{conditioned by} $a$ is the sub-observable given by $(A\mid a)_x=a\sqbrac{\iscript}A_x$ \cite{gud320}. Notice that we have
$\sum\limits _{x\in\Omega _A}(A\mid a)_x=a\sqbrac{\iscript}I=a$. If $A$ and $B$ are observables on $H$ and $\iscript$ is an instrument that measures $A$, then $B$
\textit{conditioned on} $A$ \textit{relative to} $\iscript$ is the observable \cite{gud320}
\begin{equation*}
(B\mid A)_y=\sum _{x\in\Omega _\iscript}\iscript _x^*(B_y)=\sum _{x\in\Omega _A}A_x\sqbrac{\iscript _x}B_y
\end{equation*}
If $\iscript$ and $\jscript$ are instruments on $H$ we define the instrument $\jscript$ \textit{conditioned by} $\iscript$ as \cite{gud220,gud320}
\begin{equation*}
(\jscript\mid\iscript )_y(\rho )=\sum _{x\in\Omega _\iscript}\jscript _y\paren{\iscript _x(\rho )}=\jscript _y\sqbrac{\,\iscriptbar (\rho )}
\end{equation*}
for all $\rho\in\sscript (H)$, $y\in\Omega _\jscript$. The next result corresponds to Theorem~\ref{thm12}.

\begin{thm}    % Theorem 2.1
\label{thm21}
Suppose $\iscript$ measures $A$ and $\jscript$ measures $B$.
{\rm{(i)}}\enspace $(\jscript\mid\iscript )$ measures $(B\mid A)$.
{\rm{(ii)}}\enspace For any observable $C$ we have $\paren{(C\mid B)\mid A}=\paren{C\mid (B\mid A)}$.
\end{thm}
\begin{proof}
(i)\enspace For every $\rho\in\sscript (H)$ we have
\begin{align*}  % overflow
\trace\sqbrac{(\jscript\mid\iscript )_y(\rho )}&=\trace\sqbrac{\sum _{x\in\Omega _\iscript}\jscript _y\paren{\iscript _x(\rho )}}
    =\sum _{x\in\Omega _\iscript}\trace\sqbrac{\jscript _y\paren{\iscript _x(\rho )}}=\sum _{x\in\Omega _\iscript}\trace\sqbrac{\iscript _x(\rho )B_y}\\
    &=\sum _{x\in\Omega _\iscript}\trace\sqbrac{\rho\iscript _x^*(B_y)}=\trace\sqbrac{\rho\sum _{x\in\Omega _\iscript}\iscript _x^*(B_y)}
    =\trace\sqbrac{\rho (B\mid A)_y}
\end{align*}
It follows that $(\jscript\mid\iscript )$ measures $(B\mid A)$.
(ii)\enspace It follows from (i) that $(\jscript\mid\iscript )$ measures $(B\mid A)$. Then for all $z\in\Omega _C$ we obtain
\begin{align*}
\paren{(C\mid B)\mid A}_z&=\iscriptbar ^*(C\mid B)_z=\iscriptbar ^*\sqbrac{\,\jscriptbar ^*(C_z)}=\jscriptbar ^*\circ\iscriptbar ^*(C_z)\\
   &=\overline{(\iscript\circ\jscript )}\,^*(C_z)=\overline{(\jscript\mid\iscript )^*}(C_z)=\paren{C\mid (B\mid A)}_z
\end{align*}
which gives the result.
\end{proof}

\begin{thm}    % Theorem 2.2
\label{thm22}
If $\iscript$ measures $A$ and $a\in\escript (H)$, then for all $\rho\in\sscript (H)$ we have
\begin{equation}                % equation (2.1)
\label{eq21}
\sum _{x\in\Omega _A}P_\rho (A_x)P_\rho (a\mid A_x)=P_\rho\sqbrac{(a\mid A)}=P_{\iscriptbar (\rho )}(a)
\end{equation}
\end{thm}
\begin{proof}
We have that
\begin{align*}
\sum _{x\in\Omega _A}P_\rho (A_x)P_\rho (a\mid A_x)&=\sum _{x\in\Omega _A}\trace (\rho A_x)\frac{\trace\sqbrac{\rho\iscript _x^*(a)}}{\trace (\rho A_x)}
   =\trace\sqbrac{\rho\sum _{x\in\Omega _\iscript}\iscript _x^*(a)}\\
   &=\trace\sqbrac{\rho (a\mid A)}=\trace\sqbrac{\rho\iscript _\Omega ^*(a)}=\trace\sqbrac{\,\iscriptbar (\rho )a}\\
   &=P_{\iscriptbar (\rho )}(a)
\end{align*}
and the result follows.
\end{proof}

We call \eqref{eq21} \textit{Bayes' quantum first rule}. This is the same as the classical Bayes' first rule except it depends on the instrument used to measure $A$.
We then say that \eqref{eq21} is \textit{context dependent} and that $\iscript$ is the \textit{context} in which $A$ is measured. In classical probability theory there is only one context available and no context dependence.

We say that a sub-observable $A$ is \textit{real-valued} if $\Omega _A\subseteq\real$ \cite{gud123}. If $A$ is real-valued and $\rho\in\sscript (H)$ the
$\rho$-\textit{average} (or $\rho$-\textit{expectation}) of $A$ is
\begin{equation*}
E_\rho (A)=\sum _{x\in\Omega _A}xP_\rho (A_x)=\sum _{x\in\Omega _A}x\trace (\rho A_x)
\end{equation*}
If $A$ is real-valued, we define its \textit{stochastic operator} \cite{gud123} to be the self-adjoint operator $\atilde =\sum _{x\in\Omega _A}xA _x$. We then have
\begin{equation*}
E_\rho (A)=\trace\paren{\rho\sum _{x\in\Omega _A}xA_x}=\trace (\rho\atilde\,)
\end{equation*}
which is the expectation of $\atilde$ in the state $\rho$. We also define the \textit{conditional} $\rho$-\textit{average}
\begin{equation*}
E_\rho (A\mid a)=\sum _{x\in\Omega _A}xP_\rho (A_x\mid a)=\sum _{x\in\Omega _A}\frac{x\trace\sqbrac{\rho\iscript ^*(A_x)}}{\trace (\rho a)}
\end{equation*}
where $\iscript$ measures $a$. The next result is called \textit{Bayes' quantum first rule for expectations}.

\begin{thm}    % Theorem 2.3
\label{thm23}
If $\iscript$ measures $A$ and $B$ is a real-valued observable, then
\begin{equation*}
\sum _{x\in\Omega _A}P_\rho (A_x)E_\rho (B\mid A_x)=E_\rho\sqbrac{(B\mid A)}=E_{\iscriptbar (\rho )}(B)
\end{equation*}
\end{thm}
\begin{proof}
For all $\rho\in\sscript (H)$, $x\in\Omega _A$ we have
\begin{equation*}
E_\rho (B\mid A_x)=\sum _{y\in\Omega _B}\frac{y\trace\sqbrac{\rho\iscript _x^*(B_y)}}{P_\rho (A_x)}
\end{equation*}
It follows that
\begin{align*}
\sum _{x\in\Omega _A}P_\rho (A_x)E_\rho (B\mid A_x)&=\sum _{x\in\Omega _A}\sum _{y\in\Omega _B}y\trace\sqbrac{\rho\iscript _x^*(B_y)}\\
   &=\sum _{y\in\Omega _B}y\trace\sqbrac{\rho\sum _{x\in\Omega _A}\iscript _x^*(B_y)}
   =\sum _{y\in\Omega _B}y\trace\sqbrac{\,\iscriptbar (\rho )B_y}\\
   &=E_{\iscriptbar (\rho )}(B)=\sum _{y\in\Omega _B}y\trace\sqbrac{\rho (B\mid A)_y}\\
   &=E_\rho\sqbrac{(B\mid A)}\qedhere
\end{align*}
\end{proof}

\begin{exam}{3}  % Example 3
Let $A$ be the atomic observable
\begin{equation*}
A=\brac{P_x\colon x\in\Omega _A}=\brac{\ket{\phi _x}\bra{\phi _x}\colon x\in\Omega _A}
\end{equation*}
and let $\iscript$ be the instrument
\begin{equation*}
\iscript _x(\rho )=P_x\rho P_x=\elbows{\phi _x,\rho\phi _x}\ket{\phi _x}\bra{\phi _x}
\end{equation*}
that measures $A$. Then
\begin{equation*}
A_x\sqbrac{\iscript _x}a=\iscript _x^*(a)=\elbows{\phi _x,a\phi _x}\ket{\phi _x}\bra{\phi _x}
\end{equation*}
Moreover, if $B=\brac{B_y\colon y\in\Omega _B}$ is an observable on $H$, then $(B\mid P_x)$ is the sub-observable $(B\mid P_x)_y=P_xB_yP_x$ and if
$a\in\escript (H)$, then $(a\mid A)$ is the effect $\iscript _\Omega ^*(a)$. For all $\rho\in\sscript (H)$ we obtain
\begin{align*}
P_\rho (a\mid A)&=P_\rho\sqbrac{\iscript _\Omega ^*(a)}=P_{\iscriptbar (\rho )}(a)=\trace\sqbrac{\sum _{x\in\Omega _A}\elbows{\phi _x,\rho\phi _x}P_xa}\\
   &=\sum _{x\in\Omega _A}\elbows{\phi _x,\rho\phi _x}\elbows{\phi _x,a\phi _x}
\end{align*}
If $B$ is a real-valued observable, we obtain
\begin{equation*}
E_\rho (B\mid A)=E_{\iscriptbar (\rho )}(B)=\trace\sqbrac{\,\iscriptbar (\rho )\btilde\,}=\sum _{x\in\Omega _A}\elbows{\phi _x,\rho\phi _x}\elbows{\phi _x,\btilde\phi _x}
\end{equation*}
Bayes' quantum first rule gives
\begin{equation*}
\sum _{x\in\Omega _A}P_\rho (P_x)P_\rho (a\mid P_x)=P_\rho (a\mid A)=\sum _{x\in\Omega _A}\elbows{\phi _x,\rho\phi _x}\elbows{\phi _x,a\phi _x}
\end{equation*}
and Bayes' quantum first rule for expectations gives
\begin{equation*}
\sum _{x\in\Omega _A}P_\rho (P_x)E_\rho (B\mid P_x)=E_\rho (B\mid A)=\sum _{x\in\Omega _A}\elbows{\phi _x\rho\phi _x}\elbows{\phi _x\btilde\phi _x}
\hskip 2pc\qedsymbol
\end{equation*}
\end{exam}

\begin{exam}{4}  % Example 4
If $A=\brac{A_x\colon x\in\Omega _A}$ is an observable and $\alpha _x\in\sscript (H)$, $x\in\Omega _A$, we define the \textit{Holevo instrument}
$\hscript _x^{(A,\alpha )}(\rho )=\trace (\rho A_x)\alpha _x$ \cite{hol82,hol94}. Then $\hscript ^{(A,\alpha)}$ measures $A$ because
\begin{equation*}
\trace\sqbrac{\hscript _x^{(A,\alpha )}(\rho )}=\trace\sqbrac{\trace (\rho A_x)\alpha _x}=\trace (\rho A_x)\trace (\alpha _x)=\trace (\rho A_x)
\end{equation*}
Also, the dual of $\hscript ^{(A,\alpha )}$ becomes
\begin{align*}
\hscript _x^{(A,\alpha )*}(a)&=\trace (\alpha _xa)A_x
\intertext{and}
(a\mid A)&=\hscript _{\Omega _A}^{(A,\alpha )*}(a)=\sum _{x\in\Omega _A}\hscript ^{(A,\alpha )*}(a)=\sum _{x\in\Omega _A}\trace (\alpha _xa)A_x
\end{align*}
Then Bayes' quantum first rule becomes
\begin{equation*}
\sum _{x\in\Omega _A}P_\rho (A_x)P_\rho (a\mid A_x)=P_\rho (a\mid A)=\sum _{x\in\Omega _A}\trace (\rho A_x)\trace (\alpha _xa)
\end{equation*}
Moreover, if $B=\brac{B_y\colon y\in\Omega _B}$ is a real-valued observable, then
\begin{align*}
E_\rho (B\mid A)&=\sum _{y\in\Omega _B}y\trace\sqbrac{\,\overline{\hscript ^{(A,\alpha )}}(\rho )B_y}
    =\sum _{y\in\Omega _B}y\trace\sqbrac{\sum _{x\in\Omega _A}\hscript _x^{(A,\alpha )}(\rho )B_y}\\
    &=\sum _{y\in\Omega _B}y\trace\sqbrac{\sum _{x\in\Omega _A}\trace (\rho A_x)\alpha _xB_y}=\sum _{x\in\Omega _A}\trace (\rho A_x)\trace (\alpha _x\btilde )
\end{align*}
Bayes' quantum first rule for expectations becomes
\begin{equation*}
\sum _{x\in\Omega _A}P_\rho (A_x)E_\rho (B\mid A_x)=\sum _{x\in\Omega _A}\trace (\rho A_x)\trace (\alpha _x\btilde )\hskip 6pc\qedsymbol
\end{equation*}
\end{exam}

\begin{exam}{5}  % Example 5
If $\hscript ^{(A,\alpha )}$ and $\hscript ^{(B,\beta )}$ are Holevo instruments, we show that
\begin{equation*}
\hscript ^{(A,\alpha )}\circ\hscript ^{(B,\beta )}=\hscript ^{(C,\beta )}
\end{equation*}
is the Holevo instrument with $C_{(x,y)}=\trace (\alpha _yB_y)A_x$. Indeed
\begin{align*}
\paren{\hscript ^{(A,\alpha )}\circ\hscript ^{(B,\beta )}}_{(x,y)}(\rho )&=\hscript _y^{(B,\beta )}(\hscript _x^{(A,\alpha )})(\rho )
    =\hscript _y^{(B,\beta )}\sqbrac{\trace (\rho A_x)\alpha _x}\\
    &=\trace (\rho A_x)\trace (\alpha _xB_y)\beta _y=\trace\sqbrac{\rho\trace (\alpha _xB_y)A_x}\beta _y\\
    &=\trace (\rho C_{(x,y)})\beta _y=\hscript _{(x,y)}^{(C,\beta )}(\rho )
\end{align*}
In contrast, if $\lscript ^A$, $\lscript ^B$ are L\"uders instruments $\lscript _x^{(A)}(\rho )=A_x^{1/2}\rho A_x^{1/2}$, $\lscript _y^{(B)}=B_y^{1/2}\rho B_y^{1/2}$, we show that $\lscript ^A\circ\lscript ^B$ is not L\"uders, in general. Indeed, suppose $\lscript ^A\circ\lscript ^B=\lscript ^C$. We then obtain
\begin{equation*}
(\lscript ^A\circ\lscript ^B)_{(x,y)}(\rho )=\lscript _y^B(\lscript _x^A(\rho ))=B_y^{1/2}A_x^{1/2}\rho A_x^{1/2}B_y^{1/2}=C_{(x,y)}^{1/2}\rho C_{(x,y)}^{1/2}
\end{equation*}
for all $\rho\in\sscript (H)$. Taking the trace of both sides gives $C_{(x,y)}=A_x^{1/2}B_yA_x^{1/2}$ and we conclude that
\begin{equation*}
B_y^{1/2}A_x^{1/2}\rho A_x^{1/2}B_y^{1/2}=(A_x^{1/2}B_yA_x^{1/2})^{1/2}\rho (A_x^{1/2}B_yA_x^{1/2})^{1/2}
\end{equation*}
for all $\rho\in\sscript (H)$. Letting $\rho =I/n$ where $n=\dim H$ gives
\begin{equation*}
B_y^{1/2}A_xB_y^{1/2}=A_x^{1/2}B_yA_x^{1/2}
\end{equation*}
This holds if and only if $A_xB_y=B_yA_x$, in which case $(\lscript ^A\circ\lscript ^B)_{(x,y)}=\lscript ^{A_xB_y}$ for every $x\in\Omega _A$, $y\in\Omega _B$. In a similar way, if $a,b\in\escript (H)$, then
\begin{equation*}
\hscript ^{(a,\alpha )}\circ\hscript ^{(b,\beta )}=\hscript ^{\paren{\trace (\alpha b)a,\beta}}
\end{equation*}
and $\lscript ^a\circ\lscript ^b$ is not L\"uders unless $ab=ba$ in which case $\lscript ^a\circ\lscript ^b=\lscript ^{ab}$.\hfill\qedsymbol
\end{exam}

We say that an observable $A=\brac{A_x\colon x\in\Omega _A}$ is \textit{commuting} if $A_xA_y=A_yA_x$ for all $x,y\in\Omega _A$. Also, two observables $B,C$ are
\textit{jointly commuting} if $B$ and $C$ are commuting and $B_xC_y=C_yB_x$ for all $x\in\Omega _B$, $y\in\Omega _C$.

\begin{thm}    % Theorem 2.4
\label{thm24}
Two observables $B$, $C$ are jointly commuting if and only if there exists an atomic observable $A$ and observables $B_1,C_1$, such that $B=(B_1\mid A)$,
$C=(C_1\mid A)$ relative to some instrument that measures $A$.
\end{thm}
\begin{proof}
If $B=(B_1\mid A)$, then $B_y=\sum _{x\in\Omega _A}\iscript _x^*(B_1y)$ and by Lemma~\ref{lem11}(iii) $\iscript _x^*(B_1y)=\lambda _{x,y}A_x$ for
$\lambda _{x,y}\in\sqbrac{0,1}$. Hence, $B_y=\sum _{x\in\Omega _A}\lambda _{x,y}A_x$. In a similar way, $C_z=\sum _{x\in\Omega _A}\mu _{x,z}A_x$ for
$\mu _{x,z}\in\sqbrac{0,1}$. It follows that $B$ and $C$ are jointly commuting. Conversely, if $B$ and $C$ are jointly commuting, then all the effects in
$\brac{B_y,C_z\colon y\in\Omega _B,z\in\Omega _C}$ commute so they are simultaneously diagonalizable. Hence, there exists an atomic observable $A$ such that
$B_y=\sum _{x\in\Omega _A}\trace (A_xB_y)A_x$ and $C_z=\sum _{x\in\Omega _A}\trace (A_xC_z)A_x$ for all $y\in\Omega _B$, $z\in\Omega _C$. Using the
L\"uders instrument $\lscript _x^A\rho =A_x\rho A_x$ we have
\begin{equation*}
(B\mid A)_y=\sum _{x\in\Omega _A}\lscript _x^{A^*}B_y=\sum _{x\in\Omega _A}A_xB_yA_x=\sum _{x\in\Omega _A}\trace (A_xB_y)A_x=B_y
\end{equation*}
Similarly, $(C\mid A)_z=C_z$ so $B=(B\mid A)$ and $C=(C\mid A)$.
\end{proof}

\begin{cor}    % Corollary 2.5
\label{cor25}
Observables $B,C$ are jointly commuting if and only if there exists an atomic observable $A$ such that $B=(B\mid A)$, $C=(C\mid A)$ relative to some instrument that measures $A$.
\end{cor}

A similar proof gives the following.

\begin{thm}    % Theorem 2.6
\label{thm26}
The following statements are equivalent.
{\rm{(i)}}\enspace An observable $B$ is commuting.
{\rm{(ii)}}\enspace There exists an atomic observable $A$ such that $B=(B\mid A)$.
{\rm{(iii)}}\enspace There exists an observable $C$ and an atomic observable $A$ such that $B=(C\mid A)$.
\end{thm}

\section{Uncertainty Principle and Entropy}  % Section 3
Let $B$ be a real-valued observable with stochastic operator $\btilde =\sum _{y\in\Omega _B}yB_y$. We have seen that $E_\rho (B)=\trace (\rho\btilde )$. Also, if $A$ is an arbitrary observable and the instrument $\iscript$ measures $A$, then relative to $\iscript$ we have $E_\rho (B\mid A)=\trace\sqbrac{\iscriptbar (\rho )\btilde}$. We call $E_\rho (B\mid A)$ the $\rho$-expectation of $B$ in \textit{context} $A$. If $A,B,C$ are observables and $B,C$ are real-valued, we define the
$\rho$-\textit{correlation of} $B$ \textit{and} $C$ \textit{in the context} $A$ by \cite{gud123}
\begin{align*}
\rmcor (B,C\mid A)&=\trace\sqbrac{\rho (B\mid A)^\sim (C\mid A)^\sim}-E_\rho (B\mid A)E_\rho (C\mid A)\\
    &=\trace\sqbrac{\rho (B\mid A)^\sim (C\mid A)^\sim}-\trace\sqbrac{\iscriptbar (\rho )\btilde}\trace\sqbrac{\iscriptbar (\rho )\ctilde}
\end{align*}
Although $\rmcor _\rho (B,C\mid A)$ need not be a real number, it is easy to check that
\begin{equation*}
\overline{\rmcor _\rho(B,C\mid A)}=\rmcor _\rho (C,B\mid A)
\end{equation*}
We call $\Delta _\rho (B,C\mid A)=\rmre\rmcor _\rho (B,C\mid A)$ the $\rho$-\textit{covariance of} $B$ \textit{and} $C$ \textit{in the context} $A$ \cite{gud123}.
We define the $\rho$-\textit{variance of} $B$ \textit{in the context of} $A$ \cite{gud123}
\begin{align*}
\Delta _\rho (B\mid A)&=\rmcor _\rho (B,B\mid A)=\Delta _\rho (B,B\mid A)\\
   &=\trace\brac{\rho\sqbrac{(B\mid A)^\sim}^2}-\brac{\trace\sqbrac{\iscriptbar (\rho )\btilde}}^2
\end{align*}
Defining the \textit{commutator} of $(B\mid A)^\sim$ with $(C\mid A)^\sim$ by
\begin{equation*}
\sqbrac{(B\mid A)^\sim ,(C\mid A)^\sim}=(B\mid A)^\sim (C\mid A)^\sim - (C\mid A)^\sim (B\mid A)^\sim
\end{equation*}
we obtain the \textit{uncertainty principal} \cite{gud123}:
\begin{align}                % equation (3.1)
\label{eq31}
\tfrac{1}{4}\ab{\trace\paren{\rho\sqbrac{(B\mid A)^\sim ,(C\mid A)^\sim}}}^2+\sqbrac{\Delta _\rho (B,C\mid A)}^2&=\ab{\rmcor _\rho (B,C\mid A)}^2\\
   &\le\Delta _\rho (B\mid A)\Delta _\rho (C\mid A)\notag
\end{align}
The variance $\Delta _\rho (B\mid A)$ gives the amount of uncertainty or lack of information about $B$ provided by $\rho$ relative to a first measurement of $A$. The less $\Delta _\rho (B\mid A)$ is, the more information $\rho$ provides about $B$. Equation~\eqref{eq31} gives a lower bound for the product of the uncertainties. Notice that \eqref{eq31} generalizes the usual uncertainty principle.

\begin{exam}{6}  % Example 6
Suppose $A$ is sharp in which case $A_xA_{x'}=\delta _{xx'}$ for all $x,x'\in\Omega _{A}$. Let us measure $A$ with the L\"uders instrument
$\iscript _x(\rho )=A_x\rho A_x$. We can now compute the various statistical quantities more completely. To simplify the notation we write $D_x=A_xDA_x$ for $D\in\lscript (H)$. We then have $\iscriptbar (\rho )=\sum\limits _{x\in\Omega _A}\rho _x$ and
\begin{align*}
E_\rho (B\mid A)&=\trace\sqbrac{\iscriptbar (\rho )\btilde}=\sum _{x,y}y\trace (\rho _xB_y)=\sum _x\trace (\rho _x\btilde )\\
    (B\mid A)^\sim&=\sum _yy(B\mid A)_y=\sum _x\btilde _x
\end{align*}
We then obtain
\begin{align*}
\rmcor _\rho (B,C\mid A)&=\trace\paren{\rho\sum _x\btilde _x\sum _{x'}\ctilde _{x'}}-\trace\sqbrac{\iscriptbar (\rho )\btilde}\trace\sqbrac{\iscriptbar (\rho )\ctilde}\\
   &=\sum _x\trace (\rho _x\btilde A_x\ctilde -\sum _{x,x'}\trace (\rho _x\btilde )\trace (\rho _{x'}\ctilde )
\end{align*}
\begin{align*}
\Delta _\rho (B,C\mid A)&=\rmre\rmcor _\rho (B,C\mid A)\\
   &=\tfrac{1}{2}\sqbrac{\rmcor _\rho (B,C\mid A)+\rmcor _\rho (C,B\mid A)}\\
   &=\tfrac{1}{2}\sum _x\trace\sqbrac{\rho _x(\btilde A_x\ctilde +\ctilde A_x\btilde )}-\sum _{x,x'}\trace(\rho _x\btilde )\trace (\rho _{x'}\ctilde )\\
\end{align*}
\begin{equation*}
\Delta _\rho (B\mid A)=\sum _x\trace\sqbrac{\rho (\btilde _x)^2}-\sqbrac{\sum _x\trace (\rho _x\btilde )}^2
\end{equation*}
\begin{equation*}
\Delta _\rho (C\mid A)=\sum _x\trace\sqbrac{\rho (\ctilde _x)^2}=\sqbrac{\sum _x\trace (\rho _x\ctilde )}^2
\end{equation*}
Finally, the commutator term becomes
\begin{align*}
\trace\brac{\rho\sqbrac{(B\mid A)^\sim ,(C\mid A)^\sim}}&=\trace\paren{\rho\sqbrac{\sum _x\btilde _x,\sum _{x'}\ctilde _{x'}}}\\
    &=\trace\sqbrac{\rho\sum _x(\btilde _x\ctilde _x-\ctilde _x\btilde _x)}\\
    &=\sum _x\sqbrac{\rho _x(\btilde A_x\ctilde -\ctilde A_x\btilde )}
\end{align*}
Substituting these terms into \eqref{eq31} gives the uncertainty principle for this case.\hskip 27pc\qedsymbol
\end{exam}

\begin{exam}{7}  % Example 7
Suppose $A$ is measured by the Holevo instrument $\hscript _x^{(A,\alpha )}(\rho )=\trace (\rho A_x)\alpha _x$. Then
\begin{equation*}
\overline{\hscript ^{(A,\alpha )}}(\rho )=\sum _{x\in\Omega _A}\trace (\rho A_x)\alpha _x
\end{equation*}
and we saw in Example~5 that
\begin{equation*}
E_\rho (B\mid A)=\sum _x\trace (\rho A_x)\trace (\alpha _x\btilde )
\end{equation*}
Since $\hscript _x^{(A,\alpha )*}(B_y)=\trace (\alpha _xB_y)A_x$ we obtain
\begin{align*}
(B\mid A)^\sim&=\sum _yy(B\mid A)=\sum _yy\sum _x\hscript ^{(A,x)^A}(B_y)=\sum _yy\sum _x\trace (\alpha _xB_y)A_x\\
   &=\sum _x\trace (\alpha _x\btilde )A_x
\end{align*}
It follows that
\begin{align*}
\rmcor _\rho (B,C\mid A)&=\trace\sqbrac{\rho\sum _x\trace (\alpha _x\btilde )A_x\sum _{x'}\trace (\alpha _{x'}\ctilde )A_{x'}}\\
   &\quad -\sqbrac{\sum _x\trace (\rho A_x)\trace (\alpha _x\btilde )}\sqbrac{\sum _x\trace (\rho A_x)\trace (\alpha _x\ctilde )}\\
   &=\sum _{x,x'}\trace (\alpha\btilde )\trace (\alpha _{x'}\ctilde )\sqbrac{\trace (\rho A_xA_{x'})-\trace (\rho A_x)\trace (\rho A_{x'})}
\end{align*}
\begin{align*}
\Delta _\rho&(B,C\mid A)=\rmre\rmcor _\rho (B,C\mid A)\\
    &=\sum _{x,x'}\trace (\alpha _x\btilde )\trace (\alpha _{x'}\ctilde )\brac{\tfrac{1}{2}\sqbrac{\trace\paren{\rho (A_xA_{x'}+A_{x'}A_x)}}-\trace (\rho A_x)\trace (A_{x'})}
\end{align*}
\begin{equation*}
\Delta _\rho (B\mid A)=\sum _{x,x'}\trace (\alpha _x\btilde )\trace (\alpha _{x'}\btilde )\sqbrac{\trace (\rho A_xA_{x'})-\trace (\rho A_x)\trace (\rho A_{x'})}
\end{equation*}
with a similar formula for $\Delta _\rho (C\mid A)$. Finally, the commutator term becomes
\begin{align*}
\trace\brac{\rho\sqbrac{(B\mid A)^\sim ,(C\mid A)^\sim}}&=\trace\brac{\rho\sqbrac{\sum _x\trace (\alpha _x\btilde )A_x,\sum _{x'}\trace (\alpha _{x'}\ctilde )A_{x'}}}\\
   &=\sum _{x,x'}\trace (\alpha _x\btilde )\trace (\alpha _{x'}\ctilde )\trace\paren{\rho\sqbrac{A_x,A_{x'}}}
\end{align*}
Substituting these terms into \eqref{eq31} gives the uncertainty principle for this case.\hskip 27pc\qedsymbol
\end{exam}

The uncertainty $\Delta _\rho (A)$ measures the lack of information about $A$ provided by the state $\rho$. In the dual picture, we have the lack of information $S_A(\rho )$ that a measurement of $A$ provides about the state $\rho$ and this is called entropy. We now briefly discuss conditional entropy. If $a\in\escript (H)$,
$\rho\in\sscript (H)$, we define the $\rho$-\textit{entropy} of $a$ by \cite{gud22,op04,st22,weh78}
\begin{equation*}
S_a(\rho )=-\trace (\rho a)\ln\sqbrac{\frac{\trace (\rho A)}{\trace (a)}}
\end{equation*}
We interpret $S_a (\rho )$ as the amount of uncertainty that a measurement of $a$ provides about $\rho$. The smaller $S_a(\rho )$ is, the more information a measurement of $a$ gives about $\rho$. It follows that if $\iscript$ measures $a$, then
\begin{align*}
S_{a\sqbrac{\iscript}b}(\rho )&=-\trace\paren{\rho a\sqbrac{\iscript}b}\ln\sqbrac{\frac{\trace\paren{\rho a\sqbrac{\iscript}b}}{\trace\paren{a\sqbrac{\iscript}b}}}\\
   &=-\trace\sqbrac{\rho\iscript ^*(b)}\ln\sqbrac{\frac{\trace\sqbrac{\rho\iscript ^*(b)}}{\trace\sqbrac{\iscript ^*(b)}}}
   =-\trace\sqbrac{\iscript (\rho )b}\ln\sqbrac{\frac{\trace\sqbrac{\iscript (\rho )b}}{\trace\sqbrac{\iscript ^*(b)}}}
\end{align*}
We define the $a$-\textit{conditional} $\rho$-\textit{entropy of} $b$ as
\begin{equation*}
S_{(b\mid a)}(\rho )=S_b\sqbrac{\iscript (\rho )}=-\trace\sqbrac{\iscript (\rho )b}\ln\sqbrac{\frac{\trace\sqbrac{\iscript (\rho )b}}{\trace (b)}}
\end{equation*}
Notice that there is a close connection between these two entropies. Since $ln x$ is an increasing function we have the following.

\begin{lem}    % Lemma 3.1
\label{lem31}
$S_{a\sqbrac{\iscript}b}(\rho )\le S_{(b\mid a)}(\rho )$ for every $\rho\in\sscript (H)$ if and only if $\trace\sqbrac{\iscript ^*(b)}\le\trace (b)$.
\end{lem}

\begin{exam}{8}  % Example 8
If $a$ is measured by the L\"uders operation $\lscript ^{(a)}(\rho )=a^{1/2}\rho a^{1/2}$, then
\begin{equation*}
\trace\sqbrac{\iscript ^*(b)}=\trace (a^{1/2}ba^{1/2})=\trace (ab)\le\trace (b)
\end{equation*}
so in this we have $S_{a\sqbrac{\iscript}b}(\rho )\le S_{(b\mid a)}(\rho )$ for all $\rho\in\sscript (H)$\hskip 7pc\qedsymbol
\end{exam}

\begin{exam}{9}  % Example 9
If $a$ is measured by the Holevo operation $\hscript ^{(a,\alpha )}(\rho )=\trace (\rho a)\alpha$, then
\begin{equation*}
\trace\sqbrac{\hscript ^{(a,\alpha )*}(b)}=\trace\sqbrac{\trace (\alpha b)a}=\trace (\alpha b)\trace (a)
\end{equation*}
Hence, $\trace\sqbrac{\hscript ^{(a,\alpha )*}(b)}\le\trace (b)$ if and only if $\trace (\alpha b)\trace (a)\le\trace (b)$. Depending on $a,b,\alpha$ this inequality sometimes holds and sometimes does not hold.. We conclude that $S_{a\sqbrac{\iscript}b}$ and $S_{(a\mid b)}$ give different measures of information about $\rho$.
\hskip 25pc\qedsymbol
\end{exam}

If $A$ is an observable, we define the $\rho$-\textit{entropy} of $A$ by \cite{gud22,st22}
\begin{equation*}
S_A(\rho )=\sum _{x\in\Omega _a}S_{A_x}(\rho )=-\sum _{x\in\Omega _A}\trace (\rho A_x)\ln\sqbrac{\frac{\trace (\rho A_x)}{\trace (A_x)}}
\end{equation*}
If $\iscript$ measures $A$, we define the $A$-\textit{conditional} $\rho$-\textit{entropy} of the observable $B$ by \cite{gud22,st22}
\begin{equation*}
S_{(B||A)}(\rho )=S_B\sqbrac{\,\iscriptbar (\rho )}=\sum _{y\in\Omega _B}S_{B_y}\sqbrac{\,\iscriptbar (\rho )}
\end{equation*}
As with effects, this can be compared with
\begin{equation*}
S_{(B\mid A)}(\rho )=\sum _{y\in\Omega _B}S_{(B\mid A)_y}(\rho )
\end{equation*}
and these are not related in general.

One of the advantages of $S_{(B||A)}$ over $S_{(B\mid A)}$ is the following. If $\iscript$ measures $A$ and $\jscript$ measures $B$ we obtain
\begin{equation*}
S_{\paren{(C||B)||A}}(\rho )=S_{(C||B)}\sqbrac{\,\iscriptbar (\rho )}=S_C\sqbrac{\,\jscriptbar\paren{\,\iscriptbar (\rho )}}=S_{\paren{C||(B||A)}}
\end{equation*}
but in general
\begin{equation*}
S_{\paren{(C\mid B)\mid A}}\ne S_{\paren{C\mid (B\mid A)}}
\end{equation*}
We can continue this to obtain results concerning more than three observables.

\end{document}